\documentclass[twocolumn,superscriptaddress,nofootinbib,preprintnumbers,showpacs,tightenlines,notitlepage]{revtex4}
\usepackage[latin1]{inputenc}
\usepackage{graphicx}
\usepackage{amssymb}
\usepackage{color}
\usepackage{float}
\usepackage{amsmath}
\usepackage{amsfonts}
\usepackage{dcolumn}
\usepackage{hyperref}
\usepackage{amsthm}
\usepackage{color}

\def\uudot{\dot{u}}
\def\3nab{\tilde{\nabla}}

\def\la {\langle}
\def\ra {\rangle}

\def\be {\begin{equation}}
\def\ee {\end{equation}}
\def\ba {\begin{eqnarray}}
\def\ea {\end{eqnarray}}
\newtheorem{prop}{Proposition}
\newtheorem{cor}{Corollary}
\newcommand{\bra}[1]{\left(#1\right)}
\newcommand{\bras}[1]{\left[#1\right]}
\newcommand{\brac}[1]{\left\{#1\right\}}

\newcommand{\sfr}[2]{{\textstyle\frac{#1}{#2}}}

\newcommand{\E}{{\mathcal E}}
\renewcommand{\H}{{\mathcal H}}
\newcommand{\barray}{\begin{array}}
\newcommand{\earray}{\end{array}}
\newcommand{\e}{e}
\newcommand{\N}{N}

 \newcommand{\nab}{\nabla}

\newcommand \ep {\epsilon}

\newcommand \om {\omega}

\newcommand{\udot}{{\mathcal A}}
\begin{document}
\title{Constructing black hole entropy from gravitational collapse}

\author{Giovanni Acquaviva}
 \email{acquavivag@unizulu.ac.za}
 \affiliation{Department of Mathematical Sciences, University of Zululand, Private Bag X1001, Kwa-Dlangezwa 3886, South Africa.}
\author{George F. R. Ellis}
 \email{george.ellis@uct.ac.za}
 \affiliation{Department of Mathematics and Applied Mathematics and ACGC, University of Cape Town,
Cape Town, Western Cape, South Africa.}
\author{Rituparno Goswami}
 \email{goswami@ukzn.ac.za}
\author{Aymen I. M. Hamid}
 \altaffiliation{Physics Department, University of Khartoum, Sudan.}
 \email{aymanimh@gmail.com}
 \affiliation{Astrophysics and Cosmology Research Unit, School of Mathematics, Statistics and Computer Science, University of KwaZulu-Natal, Private Bag X54001, Durban 4000, South Africa.}

\begin{abstract}
Based on a recent proposal for the gravitational entropy of free gravitational fields, we investigate the thermodynamic properties of black hole formation through gravitational collapse in the framework of the semitetrad 1+1+2 covariant formalism.  In the simplest case of an Oppenheimer-Snyder-Datt collapse we prove that the change in gravitational entropy outside a collapsing body is related to the variation of the surface area of the body itself, even before the formation of horizons.  As a result, we are able to relate the Bekenstein-Hawking entropy of the black hole endstate to the variation of the vacuum gravitational entropy outside the collapsing body.
\end{abstract}


\maketitle
\section{Introduction}
Black hole entropy in the case of eternal black holes (the maximally extended Schwarzschild vacuum solution) is a very well understood subject since the pioneering work of Bekenstein \cite{Bek73} and of Bardeen, Carter, and Hawking \cite{BarCarHaw73,Haw75}, see \cite{Pag05}  for a review. However astrophysical black holes form in a dynamic way. Entropy is not so well understood in that context. \\

In the context of astrophysical formation of  black holes, a key question arises.  We know that astrophysical black holes are not eternal in the past:  they are created by the continual gravitational collapse of massive stars. Therefore the question is,
\begin{quote}
\textsc{Question:} {\em Should black hole entropy be only a property of the black hole event  horizon, manifesting suddenly as the horizon forms, or should  it be an artefact of a time varying gravitational field due to gravitational collapse, 
with gravitational entropy changing smoothly from initial values to the canonical value $S_{BH} = A/4$ as an event horizon comes into being when the stellar surface area $r$ crosses the value $r = 2M$?}
\end{quote}

The difficulty  in working this out is that we need a definition of gravitational entropy for a generic gravitational field, not only for a black hole; but until recently we have not had such a definition. Penrose \cite{Pen79} has suggested such a definition should be based in properties of the Weyl tensor but gave no specific formula. The important idea behind this proposal is as follows: we know that for any general relativistic 
spacetime, all the information about the spacetime curvature, and hence the gravitational field, is encoded in the Riemann curvature tensor\cite{HE} . However the trace of this tensor, namely the Ricci tensor, is related pointwise to the energy momentum tensor of the matter fields via the Einstein field equations. Hence the information on entropy encoded in the Ricci tensor is the same as the entropy of the matter fields. Therefore, to characterise the entropy of the free gravitational field (or the pure gravity) apart from that encoded in the matter terms, one must use the Weyl tensor, which is the trace free part of the Riemann curvature tensor \cite{HE,EllisCovariant,Ellisbook}.\\

Recently a thermodynamically motivated measure of gravitational entropy based on this idea was proposed by Clifton et. al. \cite{Clifton:2013dha}. A strong candidate for providing a measure of the gravitational entropy of an arbitrary gravitational field is the Bel-Robinson tensor \cite{Bel-Robinson}, which is constructed from  the Weyl tensor and it's dual. It has been shown in \cite{Maartens:1997fg} that this tensor is the unique Maxwellian tensor that can be constructed from the Weyl tensor and acts like an effective energy-momentum tensor of the free gravitational field, albeit having the dimension $L^{-4}$ rather than $L^{-2}$. The proposed measure of gravitational entropy in \cite{Clifton:2013dha}, therefore, uses the square root of the Bel-Robinson tensor, which was shown to be unique for spacetimes which are of Petrov type D or N. \\

This measure of gravitational entropy for free gravitational field has all the important requirements that a measure of entropy should have. It is strictly non-negative and vanishes only for conformally flat spacetimes where the Weyl tensor is zero. It measures the local anisotropy of the free gravitational field and increases monotonically as structures forms in the early universe. Most importantly, this measure reproduces the Bekenstein-Hawking entropy of a black hole, which is the famous theorem that states that the black hole entropy at any time slice is proportional to the surface area of the black hole, which is the 2-dimensional intersection of the black hole horizon and the constant time slice \cite{Bek73,Haw75,BarCarHaw73}. Through this definition of black hole entropy, one can naturally develop the concepts of black hole thermodynamics in both classical and semiclassical regimes, leading to quantum particle creations and Hawking radiation\cite{Gibbons:1977mu}. \\

To investigate the question stated above, in the light of the gravitational entropy proposal of \cite{Clifton:2013dha}, in this paper we consider the simplest example of black hole formation by Oppenheimer-Snyder-Datt \cite{osd,datt} collapse, which describes the gravitational collapse of a spherical dustlike star immersed in a Schwarzschild vacuum. Since the exterior of the star is of Petrov type D, we can uniquely determine the entropy of the free gravitational field for a static observer even when no event horizon exists. We explicitly prove that the Bekenstein-Hawking entropy of the black hole, which is formed after an infinite time for the static observer, can be linked to the net monotonic increase in the entropy of the free gravitational field during this dynamic gravitational  collapse. This result 
relates the time varying gravitational field during the continuous gravitational collapse 
to the thermodynamic property of the final state, the black hole, where gravitational entropy is well understood \cite{Pag05}. \\

In this paper, we work on spherically symmetric black holes and use the semitetrad 1+1+2 covariant formalism for a slightly more general class of Locally Rotationally Symmetric Class (LRS) II spacetimes \cite{EllisLRS} (of which spherical symmetry is a subclass). We discuss this covariant formalism and it's usefulness in describing LRS-II spacetimes  briefly in the next two sections. We then recast the equations of the gravitational entropy in this formalism in section IV and finally prove the proposition relating the net increase in entropy to the change in the surface area of the collapsing star in section V.\\

Unless otherwise specified, we use natural units ($c=8\pi G=1$) throughout this paper, Latin indices run from 0 to 3. 
The symbol $\nabla$ represents the usual covariant derivative and $\partial$ corresponds to partial differentiation. 
We use the $(-,+,+,+)$ signature and the Riemann tensor is defined by
\begin{equation}
R^{a}{}_{bcd}=\Gamma^a{}_{bd,c}-\Gamma^a{}_{bc,d}+ \Gamma^e{}_{bd}\Gamma^a{}_{ce}-\Gamma^e{}_{bc}\Gamma^a{}_{de}\;,
\end{equation}
 The Ricci tensor is obtained by contracting the {\em first} and the {\em third} indices
\begin{equation}\label{Ricci}
R_{ab}=g^{cd}R_{cadb}\;.
\end{equation}
The Hilbert--Einstein action in the presence of matter is given by
\begin{equation}
{\cal S}=\frac12\int d^4x \sqrt{-g}\left[R-2\Lambda-2{\cal L}_m \right]\;,
\end{equation}
variation of which gives the Einstein's field equations as
\be
G_{ab}+\Lambda g_{ab}=T_{ab}\;.
\ee

\section{Semi-Tetrad covariant formalisms}
Spacetimes can be described using tetrad formalisms or metric (or coordinate) based approaches \cite{Ellisbook}. The key idea behind the semitetrad formalisms is to identify preferred directions in spacetime and to project all the geometrical quantities describing the spacetime along these preferred directions and onto the spaces perpendicular to them. Among the most used semitetrad methods are 3+1 ADM formalism (which uses a global foliation of the spacetime and hence the spacetime has to be globally hyperbolic), and the 1+3 and 1+1+2 covariant formalisms (that use a local decomposition and hence there is no constraint on the global structure). Below we briefly describe the last two formalisms

\subsection{1+3 Covariant formalism}

The 1+3 formalism developed by Ehlers, Kristian and Sachs, and Tr\"{u}mper, and summarised by Ellis \cite{EllisCovariant}\cite{Ellisbook}, is based on a local decomposition of the spacetime manifold by choosing a preferred timelike vector: all vectors and tensors are projected either along that timelike direction or on the 3-space perpendicular to it.
We define a time-like congruence with a unit tangent vector $u^a$. The natural
choice of such vector is the tangent to the matter flow lines. Any
vector $X^a$ in the manifold can then be projected on the perpendicular 3-space by the projection
tensor $h^a{}_b=g^a{}_b+u^au_b$.  We can similarly decompose the full covariant derivative of any tensorial quantity in two parts. The dot derivative ($u^a\nabla_a$) is the derivative along the timelike vector $u^a$, and the spacial derivative `$D$' is the projected derivative onto the 3-space, where the projection is done on all indices by the tensor $h^a{}_b$.
The covariant derivative of the time-like vector $u^a$ can now be
decomposed into irreducible parts as 
$\nabla_au_b=-A_au_b+\frac13h_{ab}\Theta+\sigma_{ab}+\ep_{a b
c}\om^c$, where $A_a=\dot{u_a}$ is the acceleration,
$\Theta=D_au^a$ is the expansion scalar, $\sigma_{ab}=D_{\la a}u_{b \ra}$
is the shear tensor and $w^a=\ep^{a b c}D_bu_c$ is the vorticity
vector.\\ 

Similarly the Weyl curvature tensor can be decomposed
irreducibly into the Gravito-Electric and Gravito-Magnetic parts as
$E_{ab}=C_{abcd}u^cu^d$ and $H_{ab}=\frac12\ep_{acd}C^{cd}{}_{be}u^e$
These quantities give a covariant description of tidal forces and gravitational
radiation respectively. The energy momentum tensor for a general matter field can 
be also decomposed as 
$T_{ab}=\mu u_au_b+q_au_b+q_bu_a+ph_{ab}+\pi_{ab}$
where $\mu=T_{ab}u^au^b$ is the energy density, $p=(1/3 )h^{ab}T_{ab}$ is the isotropic pressure, $q_a=q_{\la a\ra}=-h^{c}{}_aT_{cd}u^d$ is the heat flux 3-vector and $\pi_{ab}=\pi_{\la ab\ra}$ is the anisotropic stress.

\subsection{1+1+2 Covariant formalism}

As an extension to the 1+3 formalism to spacetimes having a preferred spatial direction, Clarkson and Barrett developed a 1+1+2 formalism which has been applied to spherically symmetric spacetimes \cite{Clarkson:2002jz,Betschart:2004uu,Clarkson:2007yp}.
A choice of a second preferred vector along the spatial direction $e^a$ orthogonal to $u^a$ produces another split of the spacetime: this allows any 3-vector to be irreducibly split into a scalar, which is the part of the vector parallel to $\e^a$, and a vector, lying in the 2-surface orthogonal to $\e^a$. The projection tensor $N^a{}_b\equiv h^a{}_b-e^ae_b$ projects the quantities onto these 2-surfaces orthogonal to $u^a$ and $e^a$. We will refer to these 2-sirfaces as {\em sheets}.
We can now introduce two new derivatives for any object $ \psi_{a...b}{}^{c...d} $ as natural result of the new spliting of the 3-space (for detailed discussions see \cite{Clarkson:2007yp}): 
\ba
\hat{\psi}_{a..b}{}^{c..d} &\equiv & e^{f}D_{f}\psi_{a..b}{}^{c..d}~, 
\\
\delta_f\psi_{a..b}{}^{c..d} &\equiv & N_{a}{}^{p}...N_{b}{}^qN_{r}{}^{c}..
N_{s}{}^{d}N_f{}^jD_j\psi_{p..q}{}^{r..s}\;.
\ea 
We can easily see that the {\em hat derivative} is the projection of the spatial derivative $D$ along the preferred spacelike direction and {\em  $\delta$-derivative } is the projection on the 2-sheets.
The 1+3 kinematical and Weyl quantities can be decomposed as follows: setting $\{\theta,\udot,\Omega,\Sigma,{\cal E},{\cal H},\udot^a,\Sigma^a,{\cal
E}^a,{\cal H}^a,\Sigma_{ab},{\cal E}_{ab},{\cal H}_{ab}\}$  \cite{Clarkson:2007yp}, we have
\ba
\uudot^a&=&\udot \e^a+\udot^a,\\
\omega^a&=&\Omega \e^a+\Omega^a,\\
\sigma_{ab}&=&\Sigma\bra{\e_a\e_b-\sfr{1}{2}\N_{ab}}+2\Sigma_{(a}\e_{b)}+\Sigma_{ab},\\
E_{ab}&=&{\cal E}\bra{\e_a\e_b-\sfr{1}{2}\N_{ab}}+2{\cal E}_{(a}\e_{b)}+{\cal E}_{ab},\\
H_{ab}&=&{\cal H}\bra{\e_a\e_b-\sfr{1}{2}\N_{ab}}+2{\cal H}_{(a}\e_{b)}+{\cal
H}_{ab}.
\ea
Similarly we may split the fluid variables $q^a$ and $\pi_{ab}$,
\ba
q^a&=&Q \e^a+Q^a,\\
\pi_{ab}&=&\Pi\bra{\e_a\e_b-\sfr{1}{2}\N_{ab}}+2\Pi_{(a}\e_{b)}+\Pi_{ab}.
\ea

\section{LRS-II spacetimes}
A spacetime manifold $(\mathcal{M},g)$ is called  {\it locally isotropic}, if every point $p\in (\mathcal{M},g)$ has continuous non-trivial isotropy group. When this group 
consists of spatial rotations the spacetime is called {\it locally rotationally symmetric} (LRS) \cite{EllisLRS}.  
The variables that uniquely describe an LRS spacetime are
$\brac{\udot, \Theta,\phi, \xi, \Sigma,\Omega, \E, \H, \mu, p, \Pi, Q }$.  Within the LRS class, the LRS-II admits spherically symmetric solutions, is free of rotation and is described by the variables $\brac{\udot, \Theta,\phi, \Sigma,\E, \mu, p, \Pi, Q }$, since $\Omega$, $ \xi $ and $ \H $ all vanish. These spacetimes include Schwarzschild, Robertson-Walker, Lema\^{i}tre-Tolman-Bondi (LTB), and Kottler spacetimes.\\

The most general form of the metric that describes LRS-II can be written as \cite{Stewart-Ellis}
\ba
   ds^2 &= &-A^{-2}(t,r)\,dt^2+B^2(t,r)\,dr^2 \nonumber\\
   && +C^2(t,r)\,[\,dy^2+D^2(y,k)\,dz^2\,] \;,\label{LRSds}
\ea
where $t$ and $r$ are the affine parameters along the integral curves of $u^a$ and $e^a$ respectively and $k=(1,0,-1)$ specifies the closed, flat or open geometry 
of the 2-sheets respectively. Since we are concentrating on spherically symmetric spacetimes, henceforth we will only consider $k=1$.

\subsection{The field equations}
The field equations which describes the propagation and the evolution of the geometrical covariant variables can now be found using the Ricci identities of the vectors $u^a$ and $e^a$ and the doubly contracted Bianchi identities. These are as follows:  

\textit{Propagation}:
\ba
\hat\phi  &=&-\sfr12\phi^2+\bra{\sfr13\Theta+\Sigma}\bra{\sfr23\Theta-\Sigma}
    \nonumber\\&&-\sfr23\bra{\mu+\Lambda}
    -\E -\sfr12\Pi,\,\label{hatphinl}
\\  
\hat\Sigma-\sfr23\hat\Theta&=&-\sfr32\phi\Sigma-Q\
,\label{Sigthetahat}
 \\  
\hat\E-\sfr13\hat\mu+\sfr12\hat\Pi&=&
    -\sfr32\phi\bra{\E+\sfr12\Pi}
    +\bra{\sfr12\Sigma-\sfr13\Theta}Q.\;\label{hateps}
\ea

\textit{Evolution}:
\ba
   \dot\phi &=& -\bra{\Sigma-\sfr23\Theta}\bra{\udot-\sfr12\phi}
+Q\ , \label{phidot}
\\   
\dot\Sigma-\sfr23\dot\Theta  &=&
-\udot\phi+2\bra{\sfr13\Theta-\sfr12\Sigma}^2\nonumber\\
        &&+\sfr13\bra{\mu+3p-2\Lambda}-\E+\sfr12\Pi\, ,\label{Sigthetadot}
\\  
\dot\E -\sfr13\dot \mu+\sfr12\dot\Pi &=&
    +\bra{\sfr32\Sigma-\Theta}\E
    +\sfr14\bra{\Sigma-\sfr23\Theta}\Pi\nonumber\\
    &&+\sfr12\phi Q-\sfr12\bra{\mu+p}\bra{\Sigma-\sfr23\Theta}\ . \label{edot}
\ea

\textit{Propagation/evolution}:
\ba
   \hat\udot-\dot\Theta&=&-\bra{\udot+\phi}\udot+\sfr13\Theta^2
    +\sfr32\Sigma^2 \nonumber\\
    &&+\sfr12\bra{\mu+3p-2\Lambda}\ ,\label{Raychaudhuri}
\\
    \dot\mu+\hat Q&=&-\Theta\bra{\mu+p}-\bra{\phi+2\udot}Q -
    \sfr32\Sigma\Pi,\,
\\    \label{Qhat}
\dot Q+\hat
p+\hat\Pi&=&-\bra{\sfr32\phi+\udot}\Pi-\bra{\sfr43\Theta+\Sigma} Q\nonumber\\
    &&-\bra{\mu+p}\udot\ .
\ea
The 3-Ricci scalar of the spacelike 3-space orthogonal to $u^a$ can be expressed as
\be
^3R =-2\bras{\hat \phi +\sfr34\phi^2 -K}\;,\label{3sca}
\ee
where $K$ is the Gaussian curvature of the 2-sheet and it is defined by
$^2R_{ab}=KN_{ab}$. In terms of the covariant scalars we can write the Gaussian curvature $K$ as
\be
K = \sfr13\bra{\mu+\Lambda}-\E-\sfr12\Pi +\sfr14\phi^2
-\bra{\sfr13\Theta-\sfr12\Sigma}^2\ . \label{gauss}
\ee
Finally the evolution and propagation equations for the Gaussian curvature $K$ are
\ba
\dot K &=& -\bra{\sfr23\Theta-\Sigma}K\ , \label{Kdot}\\
\hat K &=& -\phi K\ . \label{Khat}
\ea
\subsection{Misner-Sharp mass for spherically symmetric spacetimes}

In this section, we derive the Misner-Sharp \cite{Misner:1964je} mass equation for LRS-II spacetimes in terms of the 1+1+2 kinematical
quantities. This quantity represents the mass inside a 2-sphere of radius $r$ at time $t$ in terms of geometric properties on that sphere.\\

The Einstein equation for the metric (\ref{LRSds}) (assuming $k=1$ for spherical symmetry) gives \cite{Stephani:2003ika};
 \be
{\cal M}_{ms}(r,t)=\frac{1}{2\sqrt{K}}\bra{1-\nab_a C\nab^a C},
\ee
where $C$ represents the area radius of the spherical 2-sheets. In terms of the Gaussian curvature of the 2-sheets 
we obtain
\be
{\cal M}_{ms}(r,t)=\frac{1}{2\sqrt{K}}\bra{1-\frac{1}{4K^3}\nab_a K\nab^a K}.
\ee
Using the 1+1+2 decomposition of the covariant derivative for LRS-II together with (\ref{gauss}),(\ref{Kdot}),(\ref{Khat}) the Misner-Sharp
mass takes the form  
\be
{\cal M}_{ms}(r,t)=\frac{1}{2K^{3/2}}\bra{\frac{1}{3}(\mu+\Lambda)-\E-\frac{1}{2}\Pi}.
\label{mass}
\ee
We can see from the above expression that both matter and the Weyl tensor contributes to the Misner-Sharp mass. Hence even in the absence of matter (as in vacuum Schwarzschild spacetime) we have non-zero gravitational mass sourced by the Weyl curvature.  

\section{Thermodynamics of a gravitational field}
We use a thermodynamically motivated expression of the gravitational entropy measure $S_{grav}$ given in \cite{Clifton:2013dha} and based on the Bel-Robinson tensor \cite{Bel-Robinson}, which has a natural interpretation as super-energy-momentum tensor \cite{Maartens:1997fg} for the gravitational field.  In order to be well defined, the gravitational entropy has to (i) be non-negative; (ii) vanish if and only if $C_{abcd}=0$; (iii) measure the local anisotropy of free gravitational field; (iv) reduce to the Bekenstein-Hawking entropy for a Schwarzschild black hole; (v) increase as structures (inhomogeneities) form in the universe.  All these conditions are met by the definition given in \cite{Clifton:2013dha} in the cases of Coulomb-like or wave-like gravitational fields: in the following we will be interested in the former case.  This definition of gravitational entropy has been explored also in \cite{larena}, along with other proposals, in the context of LTB dust models.\\

Following \cite{Clifton:2013dha}, in order to define a thermodynamically motivated gravitational entropy one has first to assume validity of the second law of thermodynamics for a generic gravitational field, that is
\be\label{2nd}
 T_{grav}dS_{grav}=dU_{grav}+p_{grav}dV>0
\ee
where $T_{grav}$, $U_{grav}$ and $p_{grav}$ represent the effective temperature, internal energy and isotropic pressure of the free gravitational field respectively and where $V$ is the spatial volume.  Another key ingredient is the equation of energy-momentum conservation, which in 1+3 decomposition is
\be\label{conserv}
 (\rho\, v)^{\cdot} + p\, \dot{v} = v\left[ - u_aT^{ab}_{;b} - q^b_{;b}-\dot{u}^a q_a - \sigma_{ab}\pi^{ab} \right]\, .
\ee
where $v$ is a spatial volume element and the dot represents the derivative with respect to time.  By comparison between the right-hand side of eq.(\ref{2nd}) and the left-hand side of eq.(\ref{conserv}), one can define an effective thermodynamic equation
\be\label{2ndb}
 T_{grav}\dot{s}_{grav}=(\mu_{grav}\, v)^{\cdot} + p_{grav}\, \dot{v}
\ee
where $s$ is the entropy density.  The quantities on the right-hand side of eq.(\ref{2ndb}) are calculated through contractions of the effective gravitational energy-momentum tensor, defined as the square-root of the Bel-Robinson tensor.  It was shown in \cite{Clifton:2013dha} that the gravitational pressure vanishes in a Coulomb-like field ($p_{grav}=0$) and the gravitational energy density is given by
\be\label{graven}
8\pi\mu_{grav}=2\alpha \sqrt{\frac{2W}{3}}=\alpha |\E|,
\ee
where $\alpha$ is constant introduced by the definition of the gravitational energy-momentum tensor. Using (\ref{mass}) the gravitational energy density can be expressed as
\be
\mu_{grav}=\frac{\alpha}{8\pi}\left|\bra{\frac{1}{3}(\mu+\Lambda)-2{\cal M}_{ms}(r,t)K^{3/2}-\frac{1}{2}\Pi}\right|.
\label{mug}
\ee
Eq.(\ref{2ndb}) can then be written as \cite{Clifton:2013dha},
\be
T_{grav}\dot{s}_{grav}=(\mu_{grav}v)^.=-v\sigma_{ab}\bra{\Pi^{ab}_{grav}+\frac{4\pi(p+\mu)}{\sqrt{3}\alpha \sqrt{2W}}}.
\ee
This expression in 1+1+2 decomposition for LRS-II reads
\be
T_{grav}\dot{s}_{grav}=(\mu_{grav}v)^.=-v\Sigma\bra{\Pi_{grav}+\frac{8\pi(p+\mu)}{3\alpha |\E|}}.
\label{2dn1+1}
\ee
If we want to obtain the entropy of a specific gravitational configuration, we can write eq.(\ref{2ndb}) with $p_{grav}=0$ as
\begin{equation}\label{dels}
 \delta s_{grav}=\frac{\delta\bra{\mu_{grav} v}}{T_{grav}}.
\end{equation}
and integrate over a spacelike hypersurface.  The last ingredient needed is a definition for the temperature of the gravitational field.

\subsection{Temperature of the gravitational field}
We follow the proposal of \cite{Clifton:2013dha} which, in 1+3 decomposition, is given by
\be
T_{grav}\equiv\frac{|u_{a;b}l^ak^b|}{\pi}=\frac{|\dot{u}_ae^a+\frac{1}{3}\Theta+\sigma_{ab}e^ae^b|}{2\pi},
\label{temp}
\ee
where $l^a=\frac{u^a-e^a}{\sqrt{2}}$, $k^a=\frac{u^a+e^a}{\sqrt{2}}$, and $e^a$ is spacelike unit vector aligned with Weyl principal tetrad
\cite{deSwardt:2010nf}. 
The temperature in eq.(\ref{temp}) can be represented in the 1+1+2 decomposition as
\be
T_{grav}=\frac{|\udot+\frac{1}{3}\Theta+\Sigma|}{2\pi}.
\label{tem1+1}
\ee
We can not directly interpret this definition as a generalisation of Hawking and Unruh \cite{Unruh:1976db} temperatures, which are all tightly related to (and describe features of) horizons: instead we identify it with the temperature of a gravitational field as measured locally at a point of the spacetime, associated with the symmetry 2-sphere through that point.

\subsection{Gravitational entropy and structure formation}

We have already stated before, the square root of Bel-Robinson tensor being the measure of gravitational entropy,  enables structure formation naturally as the entropy increases as the structure (or inhomogeneities) forms in the universe. In this subsection we explicitly show this in terms of the covariant variables and the Misner-Sharp mass. From eq. (\ref{mass}) we get the temporal and spatial evolution of the Misner-Sharp mass as
\ba
\hat{{\cal M}}_{ms}&=&\frac{1}{4K^{3/2}}\bra{\phi(\mu+\Lambda)-\bra{\Sigma-\frac{2}{3}\Theta}Q}\label{mhat}\\
\dot{{\cal M}}_{ms}&=&\frac{1}{4K^{3/2}}\bra{\bra{\frac{2}{3}\Theta-\Sigma}\bra{\Lambda-p}-Q\phi}\label{mdot}.
\ea
Let us, for simplicity, consider the universe is filled with perfect fluid with $Q=\Pi=0$. Then using eq. (\ref{mhat}) in eq. (\ref{mug}) we get,
\be
\mu_{grav}=\frac{\alpha K^{3/2}}{6\pi\phi}\left|\bra{\hat{{\cal M}}_{ms}-\frac32\phi{\cal M}_{ms}}\right|\;.
\label{mug2}
\ee
Now from LRS-II field equations, we can easily see that for a homogeneous distribution of perfect fluid with $\hat{\mu}=\hat{p}=0$ we have $\hat{{\cal M}}_{ms}-\frac32\phi{\cal M}_{ms}=0$ on every constant time slice and hence $ \delta s_{grav}=0$. However as discussed in \cite{Clifton:2013dha}, if we start with an inhomogeneous distribution of collapsing matter (as it happens during structure formation), we have $\hat{{\cal M}}_{ms}-\frac32\phi{\cal M}_{ms}\ne0$. This will then make $d\mu_{grav}>0$ and hence $dS_{grav}>0$. Thus the thermodynamics of free gravity naturally favours structure formation, in contrast with the thermodynamics of standard matter that favours dispersion. In the light of above discussion we can predict that the vacuum gravitational entropy outside a collapsing star (integrated over each constant time slice) will increase with time, favouring the process of continual gravitational collapse. This we prove explicitly in the next section.

\section{Gravitational entropy of the vacuum around a collapsing star}
 
Having all the ingredients at hand, we want to look now at the variation of gravitational entropy outside a body which is collapsing to form a black hole in the simplest scenario, the Oppenheimer-Snyder-Datt \cite{osd, datt} dust collapse model, represented schematically in Fig. \ref{collapse}. The main features of this collapsing scenario is as follows:
\begin{enumerate}
\item The interior of the collapsing star is described by a FLRW spacetime with coordinates ($t,r,\theta,\phi$) matched with the exterior vacuum solution represented by Schwarzschild spacetime with coordinates ($\tau,R,\psi,\Phi$).
\item Though in general the entropy of the spacetime has contributions coming both from the matter and the gravitational field, but in the interior spacetime the matter entropy only contributes, since $\E=0$ for FLRW spacetimes while in the exterior vacuum only the gravitational part of the entropy does not vanish. 
\end{enumerate}
Being interested in the variation of gravitational entropy from the point of view of an external static observer, we choose to integrate eq.(\ref{dels}) over a spacelike hypersurface {\it outside} the collapsing body.  Based on the assumptions stated above, in this scenario we are able to show the following:
\begin{prop}\label{propos1}
 The increase in the instantaneous gravitational entropy outside a collapsing star during a given interval of time is proportional to the change in the surface area of the star during that interval.
\end{prop}
\begin{proof}
Outside the collapsing star the spacetime is Schwarzschild.  Therefore, by using eq.(\ref{graven}) and eq.(\ref{tem1+1}) in eq.(\ref{dels}) and integrating over a 3-volume of the exterior region at fixed time, the total entropy at a given time can be expressed as
\be\label{entropyint}
S_{grav}\equiv\int_{\sigma}\, \delta s_{grav} = \pi \alpha \int_{R(\tau_0)}^{\infty}\frac{|\E|}{\udot}\ \frac{\bar{R}^2}{\sqrt{1-\frac{2m}{\bar{R}}}} d\bar{R},
\ee
where we have used $v=u^a\eta_{abcd}dx^b dx^c dx^d$, and the timelike vector $u^a$ is given by
$u^a=\bra{1/\sqrt{|1-\frac{2m}{R}|},0,0,0}$, $R(\tau_0)$ is the radius of the collapsing star at time $\tau_0$ and $m$ is the total mass of the star.
We know that for Schwarzschild spacetime $\E$ and $\udot$ are given by \cite{Betschart:2004uu}
\be
|\E|=\frac{2m}{R^3},
\ee
\be
\udot=\frac{m}{R^2}\bra{1-\frac{2m}{R}}^{-1/2}.
\ee
Using the last two expressions in eq.(\ref{entropyint}), the gravitational entropy is then given by
\be
S_{grav}=2 \pi \alpha \int_{R(\tau_0)}^{\infty}\bar{R}d\bar{R}\ .
\ee
This is an infinite quantity, (although this can be made finite by using the idea of {\em Finite-Infinity} for a realistic astrophysical star for which spacetime is almost Minkowski at a distance of one light year). However if we calculate the change in the gravitational entropy in a time interval $(\tau-\tau_0)$ we obtain
\be
\delta S_{grav}|_{(\tau-\tau_0)}=\frac{\alpha}{4}\Big(A(\tau_0)-A(\tau)\Big),\label{deltas}
\ee
where $A(\tau)$ is the surface area of the star at any time $\tau>\tau_0$. 
\end{proof}
\begin{figure}[h!]
 \begin{center}
 \includegraphics[width=8.6cm]{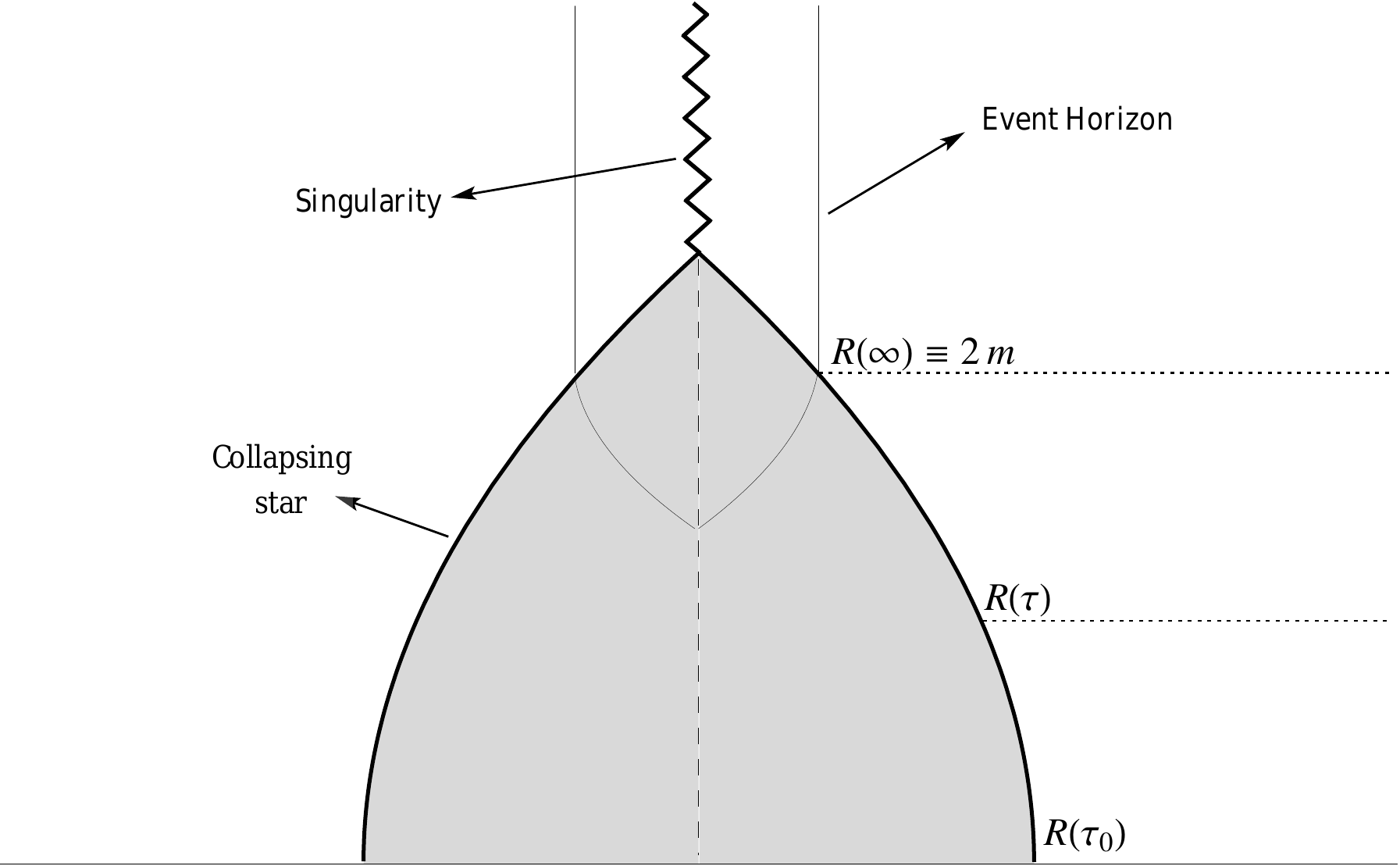}
 \caption{\label{collapse} Oppenheimer-Snyder dust collapse of a star (shaded).  In the reference frame of a static external observer, the crossing of the star's surface with the horizon at radius $2m$ occurs at $\tau\rightarrow\infty.$}
\end{center}
\end{figure}
\subsection{Constraining $\alpha$}
In order to constrain the value of the parameter $\alpha$, we can consider the variation of entropy between a configuration with $R(\tau_{\epsilon})=2m+\epsilon$ (with $\epsilon\ll2m$) and the asymptotic black hole state with $R(\infty)=2m$.  The time elapsed between these spatially neighboring states is actually infinite, because the formation of a black hole as a result of the collapse from the point of view of a static external observer is a process that takes an infinite amount of time.  Eq.(\ref{deltas}) gives
\begin{align}
 \delta S_{grav}|_{(\infty-\tau_{\epsilon})}&=\frac{\alpha}{4}\Big(A(\tau_{\epsilon})-A_H\Big)
\\
 &\simeq\alpha\, 4\pi\, m\, \epsilon\label{asymp}
\end{align}
where $\simeq$ means that we are neglecting higher orders in $\epsilon$.  The energy/mass supplied by this final stage of collapse to form the black hole is $dU\simeq\epsilon/2$, so that eq.(\ref{asymp}) can be rewritten as
\be
 dS\simeq\alpha\, (8\pi m)\ dU
\ee
The term in round brackets is the Hawking temperature $T_H=(8\pi m)^{-1}$ and hence, if $\alpha=1$, we recover the second law of black hole thermodynamics \cite{Bardeen:1973gs}.  The same value for $\alpha$ was found in \cite{Clifton:2013dha} by calculating the instantaneous gravitational entropy of a black hole and comparing the result with the known Bekenstein-Hawking value.

\subsection{Building up the Bekenstein-Hawking entropy}

As a consequence of Proposition \ref{propos1}, it is possible to show that the Bekenstein-Hawking entropy of a Schwarzschild black hole can be related to the process of collapse that leads to its formation, as stated in the following corollary

\begin{cor}
 The Bekenstein-Hawking entropy of a black hole, formed as an endstate of a spherically symmetric collapse of a massive star with Schwarzschild spacetime as the exterior, is the difference between one fourth of the initial area of the collapsing star and the net increase in the vacuum entropy in infinite collapsing time.
\end{cor}
\begin{proof}
In the reference frame of an external static observer, the crossing between the collapsing star's surface and the horizion (and hence the formation of the black hole) will take an infinite time.  Assuming that for the asymptotic black hole endstate the Bekenstein-Hawking relation $S_{BH}=\frac{1}{4}A_H$ holds, where $A_H$ is the area of the event horizon, then from Proposition \ref{propos1} with $\alpha=1$ we have
\be
S_{BH}=\frac{1}{4}A(\tau_0)- \delta S_{grav}|_{(\infty-\tau_0)}\ .
\ee
\end{proof}
The above result may have important consequences on the holographic principle \cite{susskind,hooft}, which was inspired by black hole thermodynamics. From our result above, we can easily see that the entropy being related to the surface area is not an {\it exclusive} property of horizons (black hole or cosmological), rather this property is common to other 2-surfaces enclosing a 3-volume (such as the boundary of the collapsing star). Hence, this result may expand the scope of applicability of holographic principle, which can be viewed as a manifestation of boundary value problems of the thermodynamical properties of any closed domain.

\section{Putting it all together}
The standard story of gravitational entropy relates only to black holes: it does not show how that entropy behaves as a black hole forms. But black holes form in the context of the expanding universe. The major paradox is that any standard text tells you that the second law of thermodynamics implies that entropy increases, and that in turn is taken to show that disorder increases at microscales while order increases at macro scales \cite{Pen:10}. No structure can form spontaneously. But in fact order does indeed spontaneously form on large scales as the universe expands - an apparent contradiction with the second law \cite{Ellis:95}. In order to resolve this, one needs a good definition of gravitational entropy.\\

    The definition given in \cite{Clifton:2013dha}, where (following Penrose' suggestions) gravitational entropy is based in the properties of the Weyl tensor, resolves this issue as far as the growth of perturbations in the expanding universe, due to gravitational attraction, is concerned (see equations (54) and (55) in \cite{Clifton:2013dha}). The present paper has shown that that initial growth of gravitational entropy, taking place in conjunction with the initial formation of structure in the expanding universe, can be smoothly joined on to the formation of black holes. The famous black hole entropy does not suddenly appear when the event horizon is formed; it grows steadily as gravitational attraction causes ever more concentrated objects to form, eventually leading to the existence of black holes with the standard gravitational entropy.

\begin{acknowledgments}
GA is thankful to the Astrophysics and Cosmology Research Unit (ACRU) at the University of KwaZulu-Natal for the kind hospitality.  AH,  RG and GFRE are supported by National Research Foundation (NRF), South Africa. 
\end{acknowledgments}

\end{document}